\documentclass[a4paper,11pt]{article}
\usepackage[margin=1in]{geometry}
\usepackage[T1]{fontenc}
\usepackage[utf8]{inputenc}
\usepackage[UKenglish]{babel}
\usepackage{authblk}
\usepackage{microtype}
\usepackage{amsmath,amssymb,amsthm}
\usepackage{enumerate}
\usepackage{graphicx}
\usepackage[hidelinks]{hyperref}
\usepackage[capitalise,noabbrev,nameinlink]{cleveref}
\usepackage{booktabs}
\usepackage{tikz}
\usepackage{xcolor}
\definecolor{swapearly}{RGB}{0,114,178}
\definecolor{swaplate}{RGB}{213,94,0}
\definecolor{pathearly}{RGB}{0,158,115}
\definecolor{pathlate}{RGB}{204,121,167}
\usetikzlibrary{calc,positioning,arrows.meta,decorations.pathreplacing}

\newcommand{\TG}{\mathcal{G}}
\newcommand{\texp}{\operatorname{texp}}
\newcommand{\und}{U}

\newcommand{\floor}[1]{\lfloor #1\rfloor}
\newcommand{\hub}{h}

\theoremstyle{plain}
\newtheorem{theorem}{Theorem}[section]
\newtheorem{lemma}[theorem]{Lemma}
\newtheorem{corollary}[theorem]{Corollary}
\newtheorem{proposition}[theorem]{Proposition}

\theoremstyle{definition}
\newtheorem{definition}[theorem]{Definition}

\title{Hiding a Vertex from the Temporal Explorer: A Lower Bound for Degree-Bounded Temporal Graphs}

\author{Daniele Carnevale}

\affil{University of Geneva\\ \texttt{daniele.carnevale@unige.ch}}

\date{}

\begin{document}
\maketitle

\begin{abstract}
A temporal graph is a sequence of graphs on a common set of $n$ vertices, its snapshots, one for each time step.  An agent, knowing the entire sequence in advance, may at each time step wait or move along an edge of the current snapshot, and the temporal graph is explored once the agent has visited every vertex.  We consider always-connected temporal graphs, in which every snapshot is connected, and ask how long exploration can be forced to take when the underlying graph, the union of all snapshots, has maximum degree at most~$\Delta$. Two lower bounds were known in this setting, $\Omega(\Delta n)$ and $\Omega(n\log n)$, realised by different constructions.  We establish a stronger lower bound, answering a question of  Bastide, Groenland, Michel and Rambaud. Specifically, for every $n$ and $\Delta$ with $\Delta_0\le\Delta\le n-1$ we construct an always-connected temporal graph on $n$ vertices with underlying maximum degree at most~$\Delta$ that cannot be explored in fewer than $\gamma\,\Delta n\,(1+\log(n/\Delta))$ time steps from any start vertex, where $\gamma>0$ and $\Delta_0$ are absolute constants. The construction is deterministic, and every snapshot is a spanning tree with exactly one vertex of degree greater than three. Time is divided into phases. In each phase, a rotating-cycle gadget prevents the agent from reaching more than half of the trees attached to it. Between phases, we reassign target vertices among these trees using walks on a constant-degree expander, and a Gray code order of the targets keeps the underlying degree bounded. The reassignment guarantees that, for any walk of the agent, some target remains unvisited throughout all phases.
\end{abstract}

\medskip\noindent\textbf{Keywords.} temporal graphs, temporal exploration, lower bounds, bounded degree, expander graphs.

\section{Introduction}\label{sec:intro}
Temporal graphs model networks whose topology changes over time.  A temporal graph is a
sequence $G_1,G_2,\ldots$ of graphs on a common set $V$ of $n$ vertices; the graph $G_t$, the
\emph{snapshot} at time~$t$, describes the connections available at time step~$t$.  At
each time step a temporal walk either waits at its current vertex or moves along one edge
of the current snapshot, so it can use an edge only at a time when that edge is present. The temporal
exploration problem asks for a temporal walk that starts at a given vertex and visits every
vertex of the graph as early as possible; we study the offline version, in which the whole
sequence of graphs is known in advance.

The problem was introduced by Michail and Spirakis~\cite{MS16}, whose temporal graphs may be
directed and have a finite lifetime; they proved that deciding whether a temporal graph can be
explored within its lifetime is NP-complete, and that, even when every snapshot is strongly
connected, the exploration time cannot be approximated within a factor $2-\varepsilon$ for any $\varepsilon >0$ unless
$\mathrm{P}=\mathrm{NP}$.
Much of the later work, starting with Erlebach, Hoffmann and Kammer~\cite{EHK21},
concerns undirected \emph{always-connected} temporal graphs, in which every $G_t$ is
connected; for these they showed that approximating the exploration time within
$O(n^{1-\varepsilon})$ is NP-hard for every $\varepsilon>0$.  Under this restriction an
exploration always exists: since every snapshot is connected, the set of vertices the agent can have
reached is either already all of $V$ or has an edge of the current snapshot leaving it, and
therefore grows by at least one vertex per time step.  The agent can thus reach any
prescribed vertex within $n-1$ steps, and visiting the $n$ vertices one after the other
takes $O(n^2)$ steps~\cite{MS16}. A \emph{rotating star} shows that this bound is tight in general~\cite{EHK21}. 
Half of the vertices take turns as the centre in cyclic order, and each snapshot is the star centred at the current one. 
After visiting a vertex outside this rotating set, the agent can leave it only by moving to the current centre. 
By the time it arrives, however, the centre has advanced to the next vertex. Chasing the centre repeats the same problem: 
the agent always arrives one step too late. It must therefore wait almost a full rotation for the vertex it occupies to 
become the centre before it can visit another vertex outside the rotating set. Visits to distinct vertices in this 
second set thus require $\Omega(n)$ steps each, giving an exploration time of $\Omega(n^2)$.

Since the quadratic bound is attained by very dense temporal graphs, later work asked how
much faster exploration becomes under structural restrictions on the underlying graph,
such as bounded treewidth or planarity~\cite{EHK21,AGMZ22,BGMR25,AS26}.
The line we follow bounds degrees, and here two parameters are natural.  One may bound the
maximum degree of every snapshot $G_t$, or that of the \emph{underlying graph}, the union
of all snapshots; the second is the stronger condition, since it bounds how many different
neighbours a vertex has over time.  For a bound $\Delta_s$ on the degree of every
snapshot, Erlebach and Spooner~\cite{ES18} proved an upper bound of
$O(\Delta_s\log\Delta_s\cdot n^2/\log n)$, which Erlebach, Kammer, Luo, Sajenko and
Spooner~\cite{EKLSS19} improved to $O(n^{7/4})$ for constant~$\Delta_s$. 
Bastide, Groenland, Michel and Rambaud~\cite{BGMR25} subsequently obtained a
stronger upper bound in terms of the average temporal maximum degree.

Two lower bounds were known for bounded underlying degree.  The first comes from the
rotating star itself: chaining $n/\Delta$ copies of it, each on $\Theta(\Delta)$ vertices
and sharing a leaf with the next, brings the underlying degree down to $\Delta$, while
inside a copy each vertex of the other half still costs $\Omega(\Delta)$ steps, so the
whole chain takes $\Omega(\Delta n)$ steps~\cite{EHK21}.  The second does not rely on high
degrees at all: on a planar underlying graph of maximum degree four whose snapshots are
paths~\cite{EHK21}, and even of maximum degree three~\cite{AGMZ22}, exploration can be
forced to take $\Omega(n\log n)$ steps.  Together they give $\Omega(n\max\{\Delta,\log n\})$,
but the two delays are realised by different families of temporal graphs, and Bastide
et al.~\cite{BGMR25} remark that the two constructions seem difficult to combine.
Indeed, the degree delay of the rotating star needs a vertex of large degree that moves, and the
logarithmic delay, as we obtain it, needs the vertices still to be visited to be rearranged
between phases.  A vertex whose place changes can meet new neighbours each time, and when
$\Delta$ is small compared with the number of phases this is more than the underlying
degree allows.  We show that the rearrangement can be organised so that the underlying
degree stays within~$\Delta$, and that a single temporal graph then forces both delays at
once.

\begin{theorem}\label{thm:main}
There are absolute constants $\gamma>0$ and $\Delta_0$ with the following property.  For all
integers $n$ and $\Delta$ with $\Delta_0\le \Delta\le n-1$ there is an always-connected temporal
graph $\TG$ on $n$ vertices whose underlying graph has maximum degree at most $\Delta$
and such that every exploration of $\TG$, from every start vertex, uses at least
$\gamma\,\Delta n\,(1+\log(n/\Delta))$ time steps. Every snapshot of $\TG$ is a spanning
tree with a single vertex of degree larger than three.  The temporal graph is constant from
some time step $O(\Delta n\,(1+\log(n/\Delta)))$ on, and it can be constructed
deterministically in time polynomial in~$n$.
\end{theorem}

Against the previous lower bound, for $\Delta\le n^{1-\varepsilon}$ the maximum in
$\Omega(n\max\{\Delta,\log n\})$ becomes a product.  In general the gain is a factor
$\Theta(\min\{\Delta,1+\log(n/\Delta)\})$.  It is $\Theta(\Delta)$ for $\Delta\le\log n$, so
there is none for constant $\Delta$, and $\Theta(\log n)$ for
$\log n\le\Delta\le n^{1-\varepsilon}$.  Closer to $n$ it degrades, as it must, since $n^2$
steps always suffice: it is only $\Theta(\log\log n)$ at $\Delta=n/\log n$, and constant for
$\Delta=\Theta(n)$.

Since $\Delta_s\le\Delta$, \cref{thm:main} is a lower bound for the bounded snapshot degree
setting as well.  On the upper bound side, Bastide et al.~\cite{BGMR25} showed that every
always-connected temporal graph can be explored in
$O(n^{3/2}\sqrt{\overline{\Delta}\log n})$ steps, where $\overline{\Delta}$ is its
\emph{average temporal maximum degree}, the average over the vertices $v$ of the largest
degree of $v$ over all time steps.  As $\overline{\Delta}\le\Delta_s$ (\cref{sec:model}),
their bound is $O(n^{3/2}\sqrt{\Delta_s\log n})$ for snapshot degree at most $\Delta_s$ and
$O(n^{3/2}\sqrt{\Delta\log n})$ for underlying degree at most $\Delta$. The best known upper bound for our class is therefore
$O(\min \{n^2,n^{3/2}\sqrt{\Delta\log n}\})$.
This leaves a polynomial gap whenever $\Delta$ is polynomially smaller than $n$.
For example, for any fixed $0<\varepsilon<1$, setting $\Delta=n^{1-\varepsilon}$ gives a lower bound of
$\Omega(n^{2-\varepsilon}\log n)$ and an upper bound of
$O(n^{2-\varepsilon/2}\sqrt{\log n})$.

The difficulty of combining the two known lower bounds is what led Bastide et al.\ to
Problem~3.1 of~\cite{BGMR25}, which asks whether there is a function $f=\omega(1)$ such
that for all integers $n,D\ge1$ there is a temporal graph on $n$ vertices with average
temporal maximum degree at most $D$ whose shortest exploration has length at least
$f(D)\,n\log n$. Since the average temporal maximum degree of an always-connected temporal
graph on $n$ vertices lies between $2-2/n$ and $n-1$, we consider $D$ in this range. They
note that one way to obtain such a construction would be through snapshots of maximum
degree at most~$D$.
\Cref{thm:main} answers the question positively in this way, and even with bounded
underlying degree, the most restrictive of the three degree settings.  Moreover, in every
snapshot all vertices but one have degree at most three, and the factor~$\Delta$ in the bound
comes from the remaining vertex, whose degree is $\Theta(\Delta)$ (\cref{sec:variants}).

\begin{corollary}\label{cor:problem}
Let $\Delta_0\le \Delta\le n-1$.  The temporal graph constructed for \cref{thm:main} has
average temporal maximum degree $\overline{\Delta}=\Theta(\Delta)$, and its exploration time
from every start vertex is at least $(\gamma/2)\,(\Delta/\log \Delta)\,n\log n$, which
improves to $(\gamma/2)\,\Delta\,n\log n$ when $\Delta\le\sqrt n$.
\end{corollary}

The corollary, which we prove in \cref{sec:proof} once the parameters are fixed, shows that
one can take $f(x)=\gamma x/(2\log x)$ for $x\ge\Delta_0$ in Problem~3.1: for
$\Delta_0\le D\le n-1$, the temporal graph constructed for $\Delta=D$ has average temporal
maximum degree at most $D$, and every exploration of it takes at least $f(D)\,n\log n$ steps.
The order of $f$ is best possible up to a constant factor: when $n=D+1$, the lower bound
$f(D)\,n\log n$ cannot exceed the universal $O(n^2)$ upper bound, which gives
$f(D)=O(D/\log D)$.

\subsection{The construction in outline}\label{sec:outline}

The temporal graph has three ingredients: a cycle, a set of branches hanging from one
vertex of the cycle, and a schedule that reassigns vertices among the branches from one
phase to the next.  \Cref{fig:snapshot} shows two consecutive snapshots; the precise
definitions are in \cref{sec:construction}.

\begin{figure}[t]
\centering
\begin{tikzpicture}[
  ring/.style={circle,draw,fill=white,inner sep=1.5pt},
  hubv/.style={circle,draw,fill=black,inner sep=2.1pt},
  gate/.style={rectangle,draw,fill=white,inner sep=2.1pt},
  target/.style={circle,draw,fill=black!30,inner sep=1.5pt},
  fixed/.style={circle,draw,fill=black,inner sep=1.5pt},
  every node/.style={font=\footnotesize},
  scale=0.92]
  \coordinate (C) at (-1.75,0);
  \foreach \i in {0,...,13} {
    \pgfmathsetmacro{\ang}{\i*360/14+12.857}
    \node[ring] (c\i) at ($(C)+(\ang:1.75)$) {};
  }
  \foreach \i [evaluate=\i as \j using int(\i+1)] in {1,...,12} { \draw (c\i)--(c\j); }
  \draw (c13)--(c0);
  \node[hubv] at (c0) {};
  \node[anchor=east,xshift=-5pt] at (c0) {$\hub_t$};
  \draw[-{Stealth},thick] ($(C)+(118:2.15)$) arc (118:62:2.15);
  \node[anchor=east] at ($(C)+(122:2.2)$) {rotation};
  \foreach \j/\y in {1/1.3, b/-1.3} {
    \foreach \u/\x in {1/1.0, 2/1.65, 3/2.3, 4/2.95, 5/3.6} { \node[gate] (g\j\u) at (\x,\y) {}; }
    \draw (g\j1)--(g\j2)--(g\j3)--(g\j4)--(g\j5);
    \node[fixed] (rho\j) at (4.3,\y) {};
    \draw (g\j5)--(rho\j);
    \foreach \z/\x in {1/4.95, 2/5.6, 3/6.25, 4/6.9} { \node[target] (x\j\z) at (\x,\y) {}; }
    \draw (rho\j)--(x\j1)--(x\j2)--(x\j3)--(x\j4);
    \draw (c0)--(g\j3);
    \node[anchor=west] at (7.15,\y) {branch $\j$};
  }
  \node[above=2pt of g11] {$g_{1,1}$};
  \node[above=2pt of g13] {$g_{1,u(t)}$};
  \node[above=2pt of g15] {$g_{1,p}$};
  \node[above=2pt of rho1] {$\rho_1$};
  \draw[decorate,decoration={brace,amplitude=4pt}] (0.85,1.95) -- (3.75,1.95)
    node[midway,above=4pt] {gates};
  \draw[decorate,decoration={brace,amplitude=4pt}] (4.15,1.95) -- (7.05,1.95)
    node[midway,above=4pt] {targets $X_{i,1}$};
  \node[below=2pt of gb1] {$g_{b,1}$};
  \node[below=2pt of gb3] {$g_{b,u(t)}$};
  \node[below=2pt of gb5] {$g_{b,p}$};
  \node[below=2pt of rhob] {$\rho_b$};
  \node at (4.0,0) {$\vdots$};
  \begin{scope}[shift={(0,-4.6)}]
  \coordinate (C) at (-1.75,0);
  \foreach \i in {0,...,13} {
    \pgfmathsetmacro{\ang}{\i*360/14+12.857}
    \node[ring] (c\i) at ($(C)+(\ang:1.75)$) {};
  }
  \foreach \i [evaluate=\i as \j using int(\i+1)] in {1,...,12} { \draw (c\i)--(c\j); }
  \draw (c0)--(c1);
  \node[hubv] at (c13) {};
  \node[anchor=east,xshift=-5pt] at (c13) {$\hub_{t+1}$};
    \node[anchor=east,xshift=-5pt] at (c0) {$\hub_t$};
  \foreach \j/\y in {1/1.3, b/-1.3} {
    \foreach \u/\x in {1/1.0, 2/1.65, 3/2.3, 4/2.95, 5/3.6} { \node[gate] (g\j\u) at (\x,\y) {}; }
    \draw (g\j1)--(g\j2)--(g\j3)--(g\j4)--(g\j5);
    \node[fixed] (rho\j) at (4.3,\y) {};
    \draw (g\j5)--(rho\j);
    \foreach \z/\x in {1/4.95, 2/5.6, 3/6.25, 4/6.9} { \node[target] (x\j\z) at (\x,\y) {}; }
    \draw (rho\j)--(x\j1)--(x\j2)--(x\j3)--(x\j4);
    \draw (c13)--(g\j4);
    \node[anchor=west] at (7.15,\y) {branch $\j$};
  }
  \node[above=2pt of g11] {$g_{1,1}$};
  \node[above=2pt of g14] {$g_{1,u(t+1)}$};
  \node[above=2pt of rho1] {$\rho_1$};
  \node[below=2pt of gb1] {$g_{b,1}$};
  \node[below=2pt of gb4] {$g_{b,u(t+1)}$};
  \node[below=2pt of rhob] {$\rho_b$};
  \node at (4.0,0) {$\vdots$};
  \end{scope}
\end{tikzpicture}
\caption{Two consecutive snapshots of the construction within one phase, at time $t$ (top)
and $t+1$ (bottom), with the first and the last of the $b$ branches drawn (notation as in
\cref{sec:construction}).  The cycle is present except for the edge trailing the hub.  Every
branch is attached to the hub through the gate of the hub's class, $g_{j,u(t)}$ at time $t$;
the gates before the attachment point form a dead end, and the last gate is joined to the
fixed target $\rho_j$, the first vertex of the current target path.  From one step to the
next the hub and the missing edge advance by one position, clockwise in the picture, so the
missing edge becomes the one between $\hub_t$ and $\hub_{t+1}$, and every branch is attached
at $g_{j,u(t+1)}$ instead, here the next gate along the gate path.  Nothing else changes, since
the targets are reassigned only between phases.}
\label{fig:snapshot}
\end{figure}
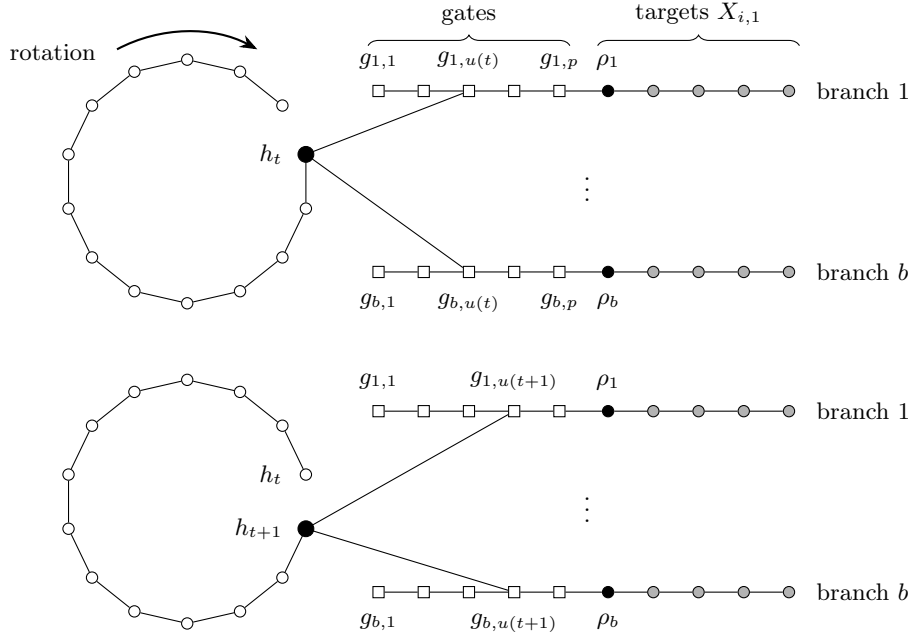
\subparagraph*{The shape of a snapshot.}
There is a cycle on $k=\Theta(n)$ vertices.  At every
time step all edges of the cycle are present except one, the missing edge, so what remains
of the cycle is a path, and one of its two endpoints is called the \emph{hub}.  There are
$b=\Theta(\Delta)$ branches.  Each branch is a path that runs through $p=\Theta(n/\Delta)$
\emph{gates} and then through $q=\Theta(n/\Delta)$ \emph{targets}, and the hub is joined to
one of its gates; no other vertex of a branch is adjacent to anything outside its branch.
Every snapshot is therefore a spanning tree made of the cycle path and the $b$ branches, each
joined to the hub by a single edge, and the hub, of degree $b+1$, is its only vertex of
degree larger than three.
\subparagraph*{The cycle as a clock.}
From one time step to the next, the missing edge advances by
one position along the cycle, and the hub, which is the endpoint of that edge in the
direction of motion, advances with it.  This is the mechanism of the rotating
star~\cite{EHK21}, transferred to a cycle.  Suppose the agent is inside some branch and
wants to reach another branch.  The only way out is through the hub, so at some step it
moves from its branch onto the hub.  At the next step, however, the missing edge has moved:
it is now exactly the edge between the vertex where the agent stands and the new hub.  The
agent is at one end of the cycle path and the hub at the other, $k-1$ edges apart.  The hub
keeps advancing, which shortens the path between them by one edge per step, and the agent
can walk towards it, which shortens the path by one more edge.  Hence the two meet only after
about $k/2$ steps, and until then the agent cannot enter any branch.  \Cref{lem:ring} makes
this precise.  Every visit to a new branch therefore costs $\Omega(k)$ time steps, so during
a phase of $L=bk/8=\Theta(\Delta n)$ steps the agent visits fewer than half of the $b$
branches.

\subparagraph*{Reassigning the targets.}
Time is divided into phases of $L$ steps, and within a
phase the snapshots differ only by the rotation, which moves the missing edge, the hub and
the $b$ edges that join the hub to the branches; the targets stay in their branches.  At
the beginning of each phase they are reassigned among the branches according to a fixed
schedule, every branch receiving $q$ targets.  In every phase the rotation leaves
more than half of the branches unvisited, and with them more than half of the targets. 
A reassignment can also put the agent in a new branch without a
move, when the target it stands on is reassigned, but this gains it at most one branch per
phase, and \cref{lem:half} allows for it.  The schedule is designed so that, whichever branches the agent leaves unvisited in each
phase, as long as they are more than half, some target stays unvisited through all $\tau$
phases, so no exploration can finish within $\tau L$ steps. 
It is built from walks on a constant-degree expander on the
branches, one walk per target, as \cref{sec:assignment} explains.  There are $b\,c^{\tau-1}$
targets for a constant $c$, one for each first branch and each choice of the $\tau-1$
steps, and since the graph has room for $\Theta(n)$ targets, this allows
$\tau=\Theta(1+\log(n/\Delta))$ phases and gives the bound
$\Omega(\Delta n\,(1+\log(n/\Delta)))$.

\subparagraph*{Bounded underlying degree.}
Two devices serve only the bound on the underlying
degree, the gates and the order of the targets; for a bound on the degree of each $G_t$
alone, the targets could hang directly from the hub (\cref{sec:variants}).  The gates are
needed because the hub moves over all $k$ cycle vertices, so a vertex joined directly to it
would acquire $k$ neighbours over time.  The cycle vertices are therefore divided into
$\Theta(n/\Delta)$ classes of $\Theta(\Delta)$ vertices each, by the residue of their index,
and each branch carries one gate per class, the hub being joined to the gate of its own
class; a gate is thus joined, over time, only to the $\Theta(\Delta)$ cycle vertices of its
class.  The order of the targets matters for a similar reason: with an arbitrary order, a
target could get new neighbours on its target path in every phase, up to $2\tau$ in all,
which is more than $\Delta$ when $\Delta$ is small compared with the number of phases.  We
order the targets by a Gray code on the step sequences of the expander walks, and then the
target paths of all phases lie in one fixed graph of maximum degree at most four.  The same ordering
might help in other constructions in which paths change from phase to phase but must stay
inside a graph of bounded degree.  Neither device helps the agent, since it
still has to pass through the hub to get from one branch to another.

\subsection{Related work}\label{sec:related}

Temporal graphs appear in the literature under several names, among them time-varying
graphs, evolving graphs and dynamic networks; the surveys~\cite{CFQS12,HS12,Mic16} give an
overview.

A line of work complementary to degree bounds studies exploration under structural
restrictions on the underlying graph. Erlebach et al.~\cite{EHK21} gave the
first subquadratic bounds, $O(n^{3/2}w^{3/2}\log n)$ for treewidth~$w$, $O(n^{9/5}\log n)$
for planar graphs and $O(n\log^3 n)$ for $2\times n$ grids, and Adamson, Gusev, Malyshev
and Zamaraev~\cite{AGMZ22} improved the first two to $O(wn^{3/2}\log n)$ and
$O(n^{7/4}\log n)$.  Both were improved further, to $O(n^{3/2}\sqrt{w\log n})$ and $O(n^{3/2}\sqrt{\log n})$,
by the general bound of Bastide et al.~\cite{BGMR25}. 
For bounded~$w$, Adamson and Spirakis~\cite{AS26} recently obtained $O(n^{4/3}w^{11/2}\log^{5/2}n)$.

Exploration is faster when the temporal graph is random.  Erlebach et al.~\cite{EHK21} showed that if every edge of an $m$-edge graph is present in every
time step independently, with probabilities that sum to at least a constant across every
cut, an exploration of $O(m\log n)$ steps exists with high probability.  Baguley
et al.~\cite{BGKSSZ26} study random spanning tree models, in which every $G_t$ is a spanning
tree drawn independently from a fixed distribution; with high probability such a model can
be explored in $O(n^{3/2})$ steps, which is tight for a suitable distribution on
stars, and in $O(m)$ steps when all its trees are spanning trees of one fixed graph with
$m$ edges.  They conjecture that $O(\sqrt{\Delta_s}\,n)$ steps suffice with high
probability when all trees have maximum degree $\Delta_s$.  For adversarial spanning trees
the dependence on the degree is at least linear instead: the chained rotating stars
of~\cite{EHK21} have snapshots of maximum degree $\Theta(\Delta)$ and force $\Omega(\Delta n)$
steps, and those of \cref{thm:main} force $\Omega(\Delta n\,(1+\log(n/\Delta)))$.

Exploration is also fast when few edges are missing: Lahtin and Zamaraev~\cite{LZ26}
showed that always-connected temporal graphs whose snapshots miss at most $\kappa$ edges of
the underlying graph can be explored in $O(\kappa n\log\kappa)$ steps for $\kappa\ge2$, improving the
$O(\kappa n\log n)$ bound of Erlebach and Spooner~\cite{ES22}, whereas our snapshots miss
$\Theta(\Delta n)$ underlying edges.  Several variants change what the agent may do: it may
make two moves per time step~\cite{EKLSS19} or traverse any number of edges of a snapshot in
one step~\cite{ES20}, it may have to move at every step~\cite{BCE23,AN25}, or, on a temporal
star, return to the centre after visiting all leaves~\cite{AMRS21}.  Others change the
objective, minimising the number of traversed edges~\cite{BST25} or asking for every edge to
be traversed~\cite{BM23}, and the parameterised complexity of exploration has also been
studied~\cite{ES23,AFGW23}.

\subsection{Organisation}

\Cref{sec:model} fixes the model.  \Cref{sec:construction} presents the construction: it
chooses the parameters, describes the snapshots and states their basic properties, taking the
assignment of targets to branches as a black box.  \Cref{sec:ring} analyses the rotating
cycle, and \cref{sec:assignment} builds the assignment.  \Cref{sec:proof} sets the admissible
range of $\Delta$, states the estimates on the parameters and proves
\cref{thm:main,cor:problem}.  \Cref{sec:variants} collects limits and variants of the
construction, and \cref{sec:conclusion} discusses open questions. \Cref{tab:notation}, after the references, collects the notation.

\section{Preliminaries}\label{sec:model}

All graphs are finite and undirected, and simple unless stated otherwise.  For a positive integer $\ell$ we write
$[\ell]=\{1,\ldots,\ell\}$.

A \emph{temporal graph} is a sequence $\TG=(G_t)_{t\ge1}$ of graphs $G_t=(V,E_t)$ on a
common finite vertex set $V$, and we write $n=|V|$.  We call $G_t$ the \emph{snapshot} at time step~$t$.  The temporal
graph is \emph{always-connected} if every $G_t$ is connected.  The \emph{underlying
graph} of $\TG$ is ${\und(\TG)=(V,\bigcup_{t\ge1}E_t)}$.

A \emph{temporal walk} of $\ell$ time steps is a sequence $W=(v_1,\ldots,v_{\ell+1})$ of
vertices such that, for every $t\in[\ell]$, either $v_{t+1}=v_t$ or $v_tv_{t+1}\in E_t$.  We
think of $v_t$ as the position of the agent at the start of time step $t$; during step $t$
the agent either waits or moves along an edge of $G_t$, and $v_{t+1}$ is its position at
the start of step $t+1$.  The walk \emph{visits} the vertices $v_1,\ldots,v_{\ell+1}$.  An
\emph{exploration} from $v^\ast$ is a temporal walk with $v_1=v^\ast$ that visits every vertex
of $V$, and the \emph{exploration time} $\texp(\TG,v^\ast)$ is the smallest number of time
steps of such a walk.  If $\TG$ is always-connected, every vertex is reachable from every position within $n-1$
steps, so $\texp(\TG,v^\ast)\le (n-1)^2$~\cite{MS16}.

We use three degree parameters.  For a vertex $v$ let $d_{\max}(v)=\max_{t\ge1}\deg_{G_t}(v)$
be its largest degree over all time steps.  We write $\Delta$ for the maximum degree of the
underlying graph $\und(\TG)$, $\Delta_s=\max_{v\in V}d_{\max}(v)$ for the largest degree of a single snapshot, and $\overline{\Delta}=\frac1n\sum_{v\in V}d_{\max}(v)$ for the average temporal
maximum degree, the parameter introduced in~\cite{BGMR25}.  Since $E_t$ is a subset of the
edge set of the underlying graph, $\deg_{G_t}(v)\le d_{\max}(v)\le\deg_{\und(\TG)}(v)$ for
all $v$ and $t$, whence $\overline{\Delta}\le \Delta_s\le \Delta$. We also use that a connected graph on $n$
vertices with exactly $n-1$ edges is a tree, to check that our snapshots are spanning
trees.

\section{The construction}\label{sec:construction}

Throughout, $d=82$ is the degree of the expander graph of \cref{lem:expander} and
$c=1+d=83$ is the number of permutations we use, the identity and the $d$ permutations
derived from that graph.  Both are fixed once and for all, and the constants $\gamma$ and
$\Delta_0$ of \cref{thm:main} depend on them.  The degree bound $\Delta$ is instead a
parameter, and in this section and in \cref{sec:proof} we assume that it is a power of two
with $\Delta\ge512$, so that $b=\Delta/64$ and $L=(\Delta/512)k$ are integers; the
general case is derived from this one at the end of \cref{sec:proof}.

\subsection{Parameters}\label{sec:parameters}

The temporal graph has three kinds of vertices: a cycle of $k$ vertices, divided into $p$
classes of at most $\Delta/2$ vertices; $b$ branches with $p$ gates each; and $bq$ targets,
$q$ for each branch.  Time is divided into $\tau$ phases of $L$ steps.  Given $n$ and $\Delta$
with $\Delta\le n-1$, these parameters are chosen as follows.  We set
\begin{equation}\label{eq:params}
  b=\frac{\Delta}{64},\qquad
  \tau=\max\bigl\{i\ge1:\ b\,c^{\,i-1}\le n/3\bigr\},\qquad
  q=c^{\,\tau-1},
\end{equation}
and then
\begin{equation}\label{eq:params2}
  p=\Bigl\lceil\frac{2n}{\Delta}\Bigr\rceil,\qquad
  k=n-bp-bq,\qquad
  L=\frac{bk}{8}=\frac{\Delta}{512}\,k.
\end{equation}
Their orders of magnitude are recorded in \cref{lem:params}; for the moment we only need
$b\ge8$, $p\ge2$, $k\ge3$ and $L\ge1$.

\subsection{Vertices}\label{sec:vertices}

The vertex set consists of three kinds of vertices.
\begin{itemize}
\item The \emph{cycle} on the vertices $r_1,\ldots,r_k$, in this cyclic order.  Indices of
  cycle vertices are read modulo~$k$, so that $r_0=r_k$ and $r_{k+1}=r_1$.  The cycle
  vertices are divided into $p$ \emph{classes} according to the residue modulo~$p$ of their
  index $z\in[k]$: the class of index $u\in[p]$ is
  \[
    S_u=\{r_z:\ z\in[k],\ z\equiv u\pmod p\}.
  \]
  Every class has $\lfloor k/p\rfloor$ or $\lceil k/p\rceil$ vertices, and
  $\lceil k/p\rceil\le\Delta/2$ because $k<n\le p\Delta/2$.
\item For every branch $j\in[b]$, the \emph{gates} $g_{j,1},\ldots,g_{j,p}$.  The gate
  $g_{j,u}$ serves the class $S_u$.
\item A set $X$ of $bq$ \emph{targets}.
\end{itemize}
In total there are $k+bp+bq=n$ vertices, by the choice of~$k$.

\subsection{The assignment of targets to branches}\label{sec:assignment-abstract}

The targets are divided among the branches, in a way that changes from phase to phase.  We
describe the requirements here and construct the assignment in \cref{sec:assignment}.

\begin{definition}\label{def:assignment}
A \emph{branch assignment} of a set $X$ of $bq$ targets with $\tau$ phases consists of one
partition of $X$ for every phase $i\in[\tau]$,
\[
  X=X_{i,1}\cup X_{i,2}\cup\cdots\cup X_{i,b},\qquad |X_{i,j}|=q\ \ (j\in[b]).
\]
The set $X_{i,j}$ is the set of targets of branch $j$ in phase $i$.
\end{definition}

\begin{definition}\label{def:evasive}
A branch assignment is \emph{evasive} if for every choice of sets
$A_1,\ldots,A_\tau\subseteq[b]$ with $|A_i|\ge b/2$ for all~$i$ there is a target $x\in X$
such that, for every $i\in[\tau]$, $x\in X_{i,j}$ for some $j\in A_i$.
\end{definition}

In the proof, $A_i$ will be the set of branches that the agent does not visit during
phase $i$.  The evasive property then gives a single target that lies in an unvisited
branch in every phase, and so is not visited during the $\tau$ phases.

Inside a branch the targets are arranged on a path.  To keep the underlying degree
bounded, we require the paths of a branch to start at the same vertex in every phase, and the
union of all paths to have bounded degree.  A \emph{path realisation} of a branch assignment
consists of vertices $\rho_1,\ldots,\rho_b\in X$ and paths $P_{i,j}$, one for every phase
$i$ and branch $j$, such that $P_{i,j}$ has vertex set $X_{i,j}$ and first vertex $\rho_j$.
In particular $\rho_j\in X_{i,j}$ for every $i$: the target $\rho_j$ never changes branch.  The
\emph{skeleton} of the realisation is the graph on $X$ whose edge set is the union of
the edge sets of all paths $P_{i,j}$.

\begin{theorem}\label{thm:assignment}
There is a constant $b_0$ with the following property.  For all integers $b\ge b_0$ and
$\tau\ge1$, with $q=c^{\tau-1}$, one can construct in time polynomial in $bq$ an
evasive branch assignment of a set of $bq$ targets with $\tau$ phases, together with a
path realisation whose skeleton has maximum degree at most four.
\end{theorem}

\subsection{Time steps}\label{sec:steps}

Fix an evasive branch assignment of $X$ with a path realisation as in
\cref{thm:assignment}, which applies because $b\ge b_0$ for every admissible $\Delta$
(\cref{sec:proof}).  Phase $i\in[\tau]$ consists of the time steps
$(i-1)L+1,\ldots,iL$, and we call the passage from one phase to the next a \emph{phase
boundary}; there are $\tau-1$ of them.  For a time step $t$ let $\hub_t=r_t$, the index
being read modulo $k$, and let $u(t)\in[p]$ be the index of its class, so that
$\hub_t\in S_{u(t)}$.  We call $\hub_t$ the \emph{hub} at time $t$.  The hub
advances by one position per step, so it travels once around the cycle every $k$
steps.

For every time step $t$ of phase $i$, the graph $G_t$ has exactly the following edges; see
\cref{fig:snapshot}.
\begin{enumerate}[(i)]
\item All edges $r_zr_{z+1}$ of the cycle except $r_{t-1}r_t$.  These
  $k-1$ edges form a Hamiltonian path of the cycle, which we always read from the hub
  outwards: the \emph{cycle path} is the sequence
  $Q_t=(\hub_t,\,r_{t+1},\,r_{t+2},\,\ldots,\,r_{t-1})$.
\item For every branch $j$, the \emph{hub edge} $\hub_t\,g_{j,u(t)}$ from the hub to the
  gate of its class.
\item For every branch $j$, the \emph{gate path} $g_{j,1}g_{j,2}\cdots g_{j,p}$ and the
  edge $g_{j,p}\rho_j$.
\item For every branch $j$, the target path $P_{i,j}$.
\end{enumerate}
\emph{Branch $j$ in phase $i$} is the vertex set $\{g_{j,1},\ldots,g_{j,p}\}\cup X_{i,j}$.
For $t>\tau L$ we set $G_t=G_{\tau L}$, so that $\TG$ is defined at all times and constant
from step $\tau L$ on, with $\hub_{\tau L}$ as its hub.

\subsection{Basic properties}

\begin{proposition}\label{prop:basic}
Every $G_t$ is a spanning tree of the vertex set.  Its only vertex of degree larger than
three is the hub, which has degree $b+1$.  The underlying graph has maximum degree at most
$\Delta/2+2$.
\end{proposition}

\begin{proof}
Fix a time step $t\le\tau L$, in phase $i$; for larger $t$ the snapshot is $G_{\tau L}$.
The path $Q_t$ has $k$ vertices and $k-1$ edges.  For every branch $j$, the hub edge, the
gate path, the edge $g_{j,p}\rho_j$ and the target path $P_{i,j}$ together form a tree on
the $p+q$ vertices of branch $j$ plus the hub, with $p+q$ edges, and this tree meets the
rest of the graph only in the hub.  Hence $G_t$ is connected and has $(k-1)+b(p+q)=n-1$
edges, so it is a spanning tree.  In $G_t$ the hub has one cycle neighbour and $b$ gate
neighbours.  A gate has at most two neighbours among the gates and $\rho_j$, and only one
gate per branch carries the hub edge, so every gate has degree at most three; $\rho_j$ has
the neighbour $g_{j,p}$ and at most one neighbour on $P_{i,j}$; all other vertices lie on
paths and have degree at most two.

For the underlying graph, consider a cycle vertex $r\in S_u$.  It has two cycle neighbours,
and it is adjacent to a gate only at times when it is the hub, when its gate neighbours
are $g_{1,u},\ldots,g_{b,u}$.  Its degree is therefore at most $b+2$.  A gate $g_{j,u}$ has at
most two neighbours among the gates and $\rho_j$, and it is joined to a cycle vertex $r$
only through a hub edge at a time when $r$ is the hub and $r\in S_u$; so its cycle
neighbours lie in $S_u$, which has at most $\lceil k/p\rceil\le\Delta/2$ vertices, and its
degree is at most $\Delta/2+2$.  Finally, a target is adjacent to other targets only through edges of the
paths $P_{i,j}$, that is, of the skeleton, which has maximum degree at most four by \cref{thm:assignment}, and its only
possible further neighbour is $g_{j,p}$, when it is $\rho_j$.  Every target therefore has
degree at most five.  All three bounds are at most $\Delta/2+2$, because $b=\Delta/64$ and
$\Delta\ge512$.
\end{proof}

\section{The rotating cycle}\label{sec:ring}

All statements concern a fixed walk $W=(v_1,v_2,\ldots)$ in the temporal graph of
\cref{sec:construction}, during its $\tau$ phases.

We say that the walk \emph{enters} a branch at step $t$ if $v_t=\hub_t$ and $v_{t+1}$ is a
gate, and that it \emph{leaves} a branch at step $t$ if $v_t$ is a gate and
$v_{t+1}=\hub_t$.  The only edges between the cycle and the rest of the graph at time $t$ are the hub edges,
so every passage between the cycle and a branch is an entry or an exit, and different
branches are joined only through the cycle.  Within a phase, therefore, the walk enters
every branch it visits, except possibly the one it is in when the phase begins.  An exit is
expensive: if the walk leaves a branch at step $t$, then at the start of step $t+1$ it
stands on the old hub $\hub_t=r_t$, the last vertex of $Q_{t+1}$, as far
from the new hub as it can be (\cref{fig:snapshot}).

\begin{lemma}\label{lem:ring}
Let $t$ be a time step and suppose that $v_t=r_{t-1}$, the last vertex of $Q_t$.
Then $v_{t'}\ne\hub_{t'}$ for every time step $t'$ with
$t\le t'<t+(k-1)/2$.
\end{lemma}

\begin{proof}
For a step $t'$ at whose start the walk is on a cycle vertex, let its offset $\ell_{t'}$ be
the position of $v_{t'}$ on $Q_{t'}$, counted from the hub, so that the walk is on the hub
exactly when $\ell_{t'}=0$, and $\ell_t=k-1$.  As long as the offset is positive, the walk
is not on the hub, so it stays on the cycle and moves at most one position along the path,
and the hub advances by one position; the offset therefore decreases by at most two.
From offset~$1$ the walk can also move onto the hub $\hub_{t'}$ itself, but then at the
start of step $t'+1$ it stands on $\hub_{t'}$, the last vertex of $Q_{t'+1}$, and its
offset is again $k-1$. In every case $\ell_{t'+1}\ge\ell_{t'}-2$, hence $\ell_{t+i}\ge k-1-2i$,
which is positive for $i<(k-1)/2$.
\end{proof}

\begin{lemma}\label{lem:entries}
Two consecutive steps at which the walk enters a branch are at least $(k+3)/2$ apart.
\end{lemma}

\begin{proof}
Let $t_1<t_3$ be two consecutive steps at which the walk enters a branch.  At the start of
step $t_1+1$ the walk is on a gate.  In every phase the hub edges are the
only edges between the cycle and the rest of the graph, so the walk returns to the cycle
only by leaving a branch, say at step $t_2\ge t_1+1$.  Then
$v_{t_2+1}=\hub_{t_2}=r_{t_2}$, and \cref{lem:ring} applied to the step
$t_2+1$ shows that the walk is not on the hub at the start of any step $t'$ with
$t_2+1\le t'<t_2+1+(k-1)/2$.  Entering a branch at step $t_3$ requires $v_{t_3}=\hub_{t_3}$,
so $t_3\ge t_2+1+(k-1)/2\ge t_1+2+(k-1)/2=t_1+(k+3)/2$.
\end{proof}

We say that the walk \emph{visits} branch $j$ during phase $i$ if $v_t$ belongs to branch
$j$ in phase $i$ for some $t$ with $(i-1)L+1\le t\le iL+1$.  The last of these positions,
$v_{iL+1}$, is also the first position of phase $i+1$, so it counts for both phases.  This
is deliberate: the targets are reassigned at the phase boundary, and the vertex on which the
walk then stands may belong to one branch under the old assignment and to another under the
new one.

\begin{lemma}\label{lem:half}
Assume $b\ge8$ and $L\le bk/8$.  During every phase the walk visits fewer than $b/2$
branches.
\end{lemma}

\begin{proof}
Fix a phase $i$.  By \cref{lem:entries} consecutive entries are more than $k/2$ steps
apart, and the phase has $L\le bk/8$ steps, so the walk enters a branch fewer than
$1+2L/k\le1+b/4$ times during the phase.  Every branch that the walk visits during the
phase either contains its first position $v_{(i-1)L+1}$ or is entered during the phase,
because within the phase the assignment does not change and the hub edges are the only
edges into a branch.  Hence the walk visits fewer than $2+b/4\le b/2$ branches, the last
inequality because $b\ge8$.
\end{proof}

\section{Evasive assignments}\label{sec:assignment}
This section proves \cref{thm:assignment}.  We first explain the idea informally, since
each ingredient of the construction is easier to follow once the difficulty it answers is
clear; the definitions and proofs follow in \cref{sec:expander,sec:labels,sec:skeleton}.

Two requirements pull against each other.  Evasion asks the targets to change branch,
since one that keeps its branch can be caught as soon as that branch is visited, and the agent
may visit almost half of the branches, different ones in every phase.  The bound on the
underlying degree asks the opposite: the targets adjacent to a target are those before and
after it on the paths $P_{i,j}$, and one whose branch changes unpredictably meets new ones
every phase.

We reconcile the two by letting the branch of a target change only to a neighbour in a
fixed expander $H$ on the $b$ branches.  From $H$ we derive $d$ permutations of the
branches: each sends every branch to a neighbour in $H$, and together they reach every neighbour
of a branch.  To these we add the identity, so that there are $c=d+1$ in all. 
Here $c$ is an absolute constant, the same for every $b$: the expander gives a fixed supply
of relabellings that serves however many branches there are.  The targets sit in a grid of
$b$ rows and $q$ columns and stay there: a phase boundary changes no vertex, only which of
them each path collects.  Every column is assigned a word of length $\tau-1$, the number of
phase boundaries, one column for each word over the alphabet of the permutations; there
are $c^{\tau-1}$ such words, so $q=c^{\tau-1}$.

In every phase every target has a label, the branch it is in.  In phase~$1$ the label of a
target is its row, so the labels of a column are all the branches, in the same order in
every column.  At the boundary from phase $i$ to phase $i+1$
the labels of a column are replaced by their images under the permutation written in
position $i$ of its word.  The common order may break
at the first boundary where two columns use different letters, but it remains true that the
labels of a column are all the $b$ branches, each exactly once: the permutation is one and the
same for the whole column, and a permutation does not merge two labels into one.  Hence in
every phase every branch has exactly one target in every column, and the path of branch $j$
in phase $i$ consists precisely of those targets, one from each column: $q$ in all, the
same for every branch and every phase.  In phase~$1$,
where all the columns are labelled alike, those paths are exactly the rows.

For a single target, the same rule says that its labels, followed over the phases, form a
walk in $H$.  The walk starts at the target's row and at the boundary from
phase $i$ to phase $i+1$ takes the step given by the permutation written in position $i$ of
its column's word, a step to a neighbour, or no step at all if that permutation is the
identity.  The walk therefore depends only on the target's row and on its column's word, not on
the paths around it.  In phase~$1$ the $q$ targets in one row are all in the same branch,
and each then follows the word of its own column.  Since the columns exhaust
all the words, every walk that starts at that branch is realised by a target in
the row, in a column whose word spells it.

What evasion needs is therefore a walk in $H$ that the agent misses.  In each phase the
branches a target could be in, if it has been in an unvisited branch in every phase so far,
form a set: at a
boundary the set spreads to the neighbours of its members, and the agent then removes the
branches it visits, at most half of all of them.  The spreading has to beat the loss, and
it does as soon as every quarter of the branches, together with its neighbours, makes up
three quarters of them; the set never shrinks below a quarter then, whatever the agent
does, a quarter being the fixed point of ``spread, then lose half''. Some walk therefore
survives all $\tau$ phases, and a target realises it, in the row
where the walk starts and in a column whose word spells it.  That target is never visited.
The grid holds $b\,c^{\tau-1}$ targets and there is room for only $\Theta(n)$, which is
what limits $\tau$ to $\Theta(1+\log(n/\Delta))$, a logarithm in base $c$; with phases of
$\Theta(\Delta n)$ steps this gives the logarithm of \cref{thm:main}.
The expansion of every quarter of the branches to three quarters is all we ask of the
$d$-regular graph $H$, and Alon's explicit expanders~\cite{Alon21} give it for $d=82$ and
every $b\ge b_0$.

What is still free is the order of the columns, and it is what keeps the underlying degree
bounded: the order does not enter the labels, so evasion is settled whatever we choose, and
all it decides is which targets are neighbours.  Between two consecutive columns the paths
pair targets carrying the same label, one from each column.  This pairing survives any
boundary at which the two columns are relabelled alike, so it can change only at the
boundaries where their words differ.  A Gray code orders the words so that consecutive ones differ in a single letter;
two consecutive columns are then joined by at most two pairings in the whole lifetime of the graph, and a target never
has more than two neighbours on each side (\cref{sec:skeleton}).

\subsection{Expanders and permutations}\label{sec:expander}

We state the property of $H$ in the form in which we use it, and prove it from an explicit
construction of Alon~\cite{Alon21} and the standard argument that a spectral gap forces
large sets to have large neighbourhoods.

\begin{lemma}\label{lem:expander}
There is a constant $b_0$ such that for every $b\ge b_0$ one can construct in time
polynomial in $b$ a $d$-regular graph $H$ on $[b]$, not necessarily simple, in which every
$Y\subseteq[b]$ with $|Y|\ge b/4$ has, together with its neighbours, at least $3b/4$
elements.  Moreover one can construct in polynomial time permutations
$\pi_1,\ldots,\pi_d$ of $[b]$ such that for every $x\in[b]$ the neighbourhood of $x$ in
$H$ is $\{\pi_1(x),\ldots,\pi_d(x)\}$.
\end{lemma}
\begin{proof}
Alon~\cite[Theorem~1.3]{Alon21} shows that for every degree $d$, every $\varepsilon>0$ and
every sufficiently large $b$ with $bd$ even, one can construct in time polynomial in $b$ a
$d$-regular graph on $b$ vertices whose eigenvalues other than the largest, $d$, have
absolute value at most $2\sqrt{d-1}+\varepsilon$.  Our $d=82$ is even, so $bd$ is even for
every $b$.  We take $\varepsilon=1$, so that these eigenvalues are at most
$2\sqrt{81}+1=19\le d/4$ in absolute value, let $b_0$ be the threshold from which the theorem
applies for $d=82$ and $\varepsilon=1$, and let $H$ be the resulting graph on $[b]$.

For the expansion, let $Y\subseteq[b]$ with $|Y|\ge b/4$, and let $Z$ be the set of
vertices neither in $Y$ nor adjacent to a vertex of $Y$.  There are no edges
between $Y$ and $Z$, so the expander mixing lemma~\cite[Lemma~2.5]{HLW06} gives
$d|Y||Z|/b\le(d/4)\sqrt{|Y||Z|}$, that is, $|Y||Z|\le b^2/16$.  Hence $|Z|\le b/4$, and $Y$
together with its neighbours has $b-|Z|\ge3b/4$ elements.

For the permutations, let $H'$ be the bipartite multigraph on $[b]\times\{1,2\}$ with as
many edges between $(x,1)$ and $(y,2)$ as $H$ has between $x$ and $y$, for all
$x,y\in[b]$.  It is $d$-regular (a loop of $H$ counts once in the degree, as
in~\cite{HLW06}), so by K\H{o}nig's theorem it is the disjoint union of $d$ perfect
matchings, which can be found in polynomial time~\cite{LP86}.  In the $i$-th of these
matchings every vertex $(x,1)$ has exactly one partner $(y,2)$, and we set $\pi_i(x)=y$;
since the matching is perfect, $\pi_i$ is a permutation of $[b]$.  Every edge of $H'$ lies
in exactly one of the matchings, so each $y$ appears among $\pi_1(x),\ldots,\pi_d(x)$ as
many times as $H$ has edges between $x$ and $y$, and these are therefore the neighbours
of $x$.
\end{proof}

We fix $H$ and $\pi_1,\ldots,\pi_d$ from \cref{lem:expander}, and we let $\pi_0$ be the
identity permutation, so that a target may also stay in its branch from one phase to the
next.  We thus have $c=d+1$ permutations $\pi_0,\ldots,\pi_d$ of $[b]$, and for every
$Y\subseteq[b]$ the images of $Y$ under them are exactly $Y$ and its neighbours, the set
whose size \cref{lem:expander} bounds from below.

\subsection{Words, columns and labels}\label{sec:labels}

Let $\tau\ge1$ and $q=c^{\tau-1}$.  A \emph{word} is a sequence
$w=(w_1,\ldots,w_{\tau-1})$ with entries in $\{0,\ldots,d\}$; there are $q$ words.  (For
$\tau=1$ there is a single empty word and everything is trivial: there is one column,
one target per branch, and no skeleton edges.)  Let $X$ be a set of $bq$ targets, and name
its elements $x_{w,a}$, one for every word $w$ and every $a\in[b]$.  The $b$ targets with
the same word form the \emph{column} $C_w$, so $X$ is the union of the $q$ columns, and we
call $a$ the \emph{row} of $x_{w,a}$.

The branch of a target in phase $i$ is determined by following the permutations named in
its word. For a word $w$, a row $a\in[b]$ and a phase $i\in[\tau]$ let
\begin{equation}\label{eq:label}
  \lambda_i(w,a)=\pi_{w_{i-1}}\circ\pi_{w_{i-2}}\circ\cdots\circ\pi_{w_1}(a),
\end{equation}
where the composition is empty for $i=1$, so that $\lambda_1(w,a)=a$ for every word, and
$\lambda_{i+1}(w,a)=\pi_{w_i}(\lambda_i(w,a))$ for $i<\tau$.  We call
$\lambda_i(w,a)$ the \emph{label} of $x_{w,a}$ in phase $i$, and we put
$X_{i,j}=\{x_{w,a}:\ \lambda_i(w,a)=j\}$.  For fixed $w$ and $i$ the map
$a\mapsto\lambda_i(w,a)$ is a composition of permutations, hence a permutation of $[b]$, so
every column contains exactly one target of each label in each phase and $|X_{i,j}|=q$ for
all $i$ and $j$.  Thus the sets $X_{i,j}$ form a branch assignment in the sense of
\cref{def:assignment}.

The labels of $x_{w,a}$ form the walk in $H$ that starts at $a$ and is named by $w$; when
$w_i=0$ the label is the same in phases $i$ and $i+1$.  \Cref{fig:skeleton} shows the
columns of a small example, whose other features are explained in \cref{sec:skeleton}.

\begin{lemma}\label{lem:evasive}
The branch assignment $(X_{i,j})$ is evasive.
\end{lemma}

\begin{proof}
Let $A_1,\ldots,A_\tau\subseteq[b]$ with $|A_i|\ge b/2$ for every $i$.  For $Y\subseteq[b]$
write $N[Y]$ for $Y$ together with its neighbours in $H$.  Define $Y_1=A_1$ and
$Y_{i+1}=N[Y_i]\cap A_{i+1}$ for $1\le i<\tau$, so that $Y_i\subseteq A_i$ for every $i$.
We claim that $|Y_i|\ge b/4$ for every $i$.  This holds for $i=1$, as $|A_1|\ge b/2$.  If
it holds for $i$, then $|N[Y_i]|\ge3b/4$ by \cref{lem:expander}, and since
$|A_{i+1}|\ge b/2$ the set $Y_{i+1}=N[Y_i]\cap A_{i+1}$ has at least $3b/4+b/2-b=b/4$
elements.

In particular $Y_\tau$ is not empty, so choose $y_\tau\in Y_\tau$.  For
$i=\tau-1,\tau-2,\ldots,1$ in turn, $y_{i+1}$ lies in $Y_{i+1}\subseteq N[Y_i]$, so there
is $y_i\in Y_i$ with $y_{i+1}\in\{\pi_0(y_i),\ldots,\pi_d(y_i)\}$; choose such a $y_i$ and
let $w_i\in\{0,\ldots,d\}$ be an index with $y_{i+1}=\pi_{w_i}(y_i)$.  By \eqref{eq:label}
the target $x_{w,y_1}$ with $w=(w_1,\ldots,w_{\tau-1})$ has label $\lambda_i(w,y_1)=y_i$ in
phase $i$, so it lies in $X_{i,y_i}$ with $y_i\in A_i$, for every $i$.
\end{proof}

\begin{figure}[t]
\centering
\begin{tikzpicture}[every node/.style={font=\footnotesize},
  tg/.style={circle,draw,fill=black!25,inner sep=1.4pt}]
  \node[font=\scriptsize] at (0.00,0.48) {$00$};
  \node[circle,draw,fill=black,inner sep=1.4pt] (c01) at (0.00,0.00) {};
  \node[circle,draw,fill=black,inner sep=1.4pt] (c02) at (0.00,-0.62) {};
  \node[circle,draw,fill=black,inner sep=1.4pt] (c03) at (0.00,-1.24) {};
  \node[circle,draw,fill=black,inner sep=1.4pt] (c04) at (0.00,-1.86) {};
  \node[font=\scriptsize] at (1.32,0.48) {$01$};
  \node[tg] (c11) at (1.32,0.00) {};
  \node[tg] (c12) at (1.32,-0.62) {};
  \node[tg] (c13) at (1.32,-1.24) {};
  \node[tg] (c14) at (1.32,-1.86) {};
  \node[font=\scriptsize] at (2.64,0.48) {$02$};
  \node[tg] (c21) at (2.64,0.00) {};
  \node[tg] (c22) at (2.64,-0.62) {};
  \node[tg] (c23) at (2.64,-1.24) {};
  \node[tg] (c24) at (2.64,-1.86) {};
  \node[font=\scriptsize] at (3.96,0.48) {$12$};
  \node[tg] (c31) at (3.96,0.00) {};
  \node[tg] (c32) at (3.96,-0.62) {};
  \node[tg] (c33) at (3.96,-1.24) {};
  \node[tg] (c34) at (3.96,-1.86) {};
  \node[font=\scriptsize] at (5.28,0.48) {$11$};
  \node[tg] (c41) at (5.28,0.00) {};
  \node[tg] (c42) at (5.28,-0.62) {};
  \node[tg] (c43) at (5.28,-1.24) {};
  \node[tg] (c44) at (5.28,-1.86) {};
  \node[font=\scriptsize] at (6.60,0.48) {$10$};
  \node[tg] (c51) at (6.60,0.00) {};
  \node[tg] (c52) at (6.60,-0.62) {};
  \node[tg] (c53) at (6.60,-1.24) {};
  \node[tg] (c54) at (6.60,-1.86) {};
  \node[font=\scriptsize] at (7.92,0.48) {$20$};
  \node[tg] (c61) at (7.92,0.00) {};
  \node[tg] (c62) at (7.92,-0.62) {};
  \node[tg] (c63) at (7.92,-1.24) {};
  \node[tg] (c64) at (7.92,-1.86) {};
  \node[font=\scriptsize] at (9.24,0.48) {$21$};
  \node[tg] (c71) at (9.24,0.00) {};
  \node[tg] (c72) at (9.24,-0.62) {};
  \node[tg] (c73) at (9.24,-1.24) {};
  \node[tg] (c74) at (9.24,-1.86) {};
  \node[font=\scriptsize] at (10.56,0.48) {$22$};
  \node[tg] (c81) at (10.56,0.00) {};
  \node[tg] (c82) at (10.56,-0.62) {};
  \node[tg] (c83) at (10.56,-1.24) {};
  \node[tg] (c84) at (10.56,-1.86) {};
  \node[anchor=east,font=\scriptsize] at (-0.22,0.00) {$a=1$};
  \node[anchor=east,font=\scriptsize] at (-0.22,-0.62) {$a=2$};
  \node[anchor=east,font=\scriptsize] at (-0.22,-1.24) {$a=3$};
  \node[anchor=east,font=\scriptsize] at (-0.22,-1.86) {$a=4$};
  \draw[line width=4.2pt,pathearly,opacity=0.42,line cap=round] (c01)--(c11)--(c21)--(c32)--(c42)--(c52)--(c63)--(c73)--(c83);
  \draw[line width=4.2pt,pathlate,opacity=0.5,line cap=round] (c01)--(c12)--(c23)--(c34)--(c41)--(c52)--(c63)--(c74)--(c81);
  \draw (c01)--(c11);
  \draw[swaplate,dash pattern=on 3pt off 1.4pt on 1pt off 1.4pt,semithick] (c01)--(c12);
  \draw (c02)--(c12);
  \draw[swaplate,dash pattern=on 3pt off 1.4pt on 1pt off 1.4pt,semithick] (c02)--(c11);
  \draw (c03)--(c13);
  \draw[swaplate,dash pattern=on 3pt off 1.4pt on 1pt off 1.4pt,semithick] (c03)--(c14);
  \draw (c04)--(c14);
  \draw[swaplate,dash pattern=on 3pt off 1.4pt on 1pt off 1.4pt,semithick] (c04)--(c13);
  \draw (c11)--(c21);
  \draw[swaplate,dash pattern=on 3pt off 1.4pt on 1pt off 1.4pt,semithick] (c11)--(c24);
  \draw (c12)--(c22);
  \draw[swaplate,dash pattern=on 3pt off 1.4pt on 1pt off 1.4pt,semithick] (c12)--(c23);
  \draw (c13)--(c23);
  \draw[swaplate,dash pattern=on 3pt off 1.4pt on 1pt off 1.4pt,semithick] (c13)--(c22);
  \draw (c14)--(c24);
  \draw[swaplate,dash pattern=on 3pt off 1.4pt on 1pt off 1.4pt,semithick] (c14)--(c21);
  \draw (c21)--(c31);
  \draw[swapearly,dashed,semithick] (c21)--(c32);
  \draw (c22)--(c32);
  \draw[swapearly,dashed,semithick] (c22)--(c31);
  \draw (c23)--(c33);
  \draw[swapearly,dashed,semithick] (c23)--(c34);
  \draw (c24)--(c34);
  \draw[swapearly,dashed,semithick] (c24)--(c33);
  \draw (c31)--(c41);
  \draw[swaplate,dash pattern=on 3pt off 1.4pt on 1pt off 1.4pt,semithick] (c31)--(c44);
  \draw (c32)--(c42);
  \draw[swaplate,dash pattern=on 3pt off 1.4pt on 1pt off 1.4pt,semithick] (c32)--(c43);
  \draw (c33)--(c43);
  \draw[swaplate,dash pattern=on 3pt off 1.4pt on 1pt off 1.4pt,semithick] (c33)--(c42);
  \draw (c34)--(c44);
  \draw[swaplate,dash pattern=on 3pt off 1.4pt on 1pt off 1.4pt,semithick] (c34)--(c41);
  \draw (c41)--(c51);
  \draw[swaplate,dash pattern=on 3pt off 1.4pt on 1pt off 1.4pt,semithick] (c41)--(c52);
  \draw (c42)--(c52);
  \draw[swaplate,dash pattern=on 3pt off 1.4pt on 1pt off 1.4pt,semithick] (c42)--(c51);
  \draw (c43)--(c53);
  \draw[swaplate,dash pattern=on 3pt off 1.4pt on 1pt off 1.4pt,semithick] (c43)--(c54);
  \draw (c44)--(c54);
  \draw[swaplate,dash pattern=on 3pt off 1.4pt on 1pt off 1.4pt,semithick] (c44)--(c53);
  \draw (c51)--(c61);
  \draw[swapearly,dashed,semithick] (c51)--(c64);
  \draw (c52)--(c62);
  \draw[swapearly,dashed,semithick] (c52)--(c63);
  \draw (c53)--(c63);
  \draw[swapearly,dashed,semithick] (c53)--(c62);
  \draw (c54)--(c64);
  \draw[swapearly,dashed,semithick] (c54)--(c61);
  \draw (c61)--(c71);
  \draw[swaplate,dash pattern=on 3pt off 1.4pt on 1pt off 1.4pt,semithick] (c61)--(c72);
  \draw (c62)--(c72);
  \draw[swaplate,dash pattern=on 3pt off 1.4pt on 1pt off 1.4pt,semithick] (c62)--(c71);
  \draw (c63)--(c73);
  \draw[swaplate,dash pattern=on 3pt off 1.4pt on 1pt off 1.4pt,semithick] (c63)--(c74);
  \draw (c64)--(c74);
  \draw[swaplate,dash pattern=on 3pt off 1.4pt on 1pt off 1.4pt,semithick] (c64)--(c73);
  \draw (c71)--(c81);
  \draw[swaplate,dash pattern=on 3pt off 1.4pt on 1pt off 1.4pt,semithick] (c71)--(c84);
  \draw (c72)--(c82);
  \draw[swaplate,dash pattern=on 3pt off 1.4pt on 1pt off 1.4pt,semithick] (c72)--(c83);
  \draw (c73)--(c83);
  \draw[swaplate,dash pattern=on 3pt off 1.4pt on 1pt off 1.4pt,semithick] (c73)--(c82);
  \draw (c74)--(c84);
  \draw[swaplate,dash pattern=on 3pt off 1.4pt on 1pt off 1.4pt,semithick] (c74)--(c81);
  \begin{scope}[every node/.style={font=\scriptsize,inner sep=0pt,anchor=base west}]
    \def\legendgap{15pt}
    \node (lg0) at (0,-2.46cm-2pt) {matching used from phase};
    \draw[swapearly,dashed,semithick]
      ([xshift=\legendgap,yshift=2pt]lg0.base east) -- ++(15pt,0) coordinate (le1);
    \node (lg1) at ([xshift=3pt,yshift=-2pt]le1) {$2$};
    \draw[swaplate,dash pattern=on 3pt off 1.4pt on 1pt off 1.4pt,semithick]
      ([xshift=\legendgap,yshift=2pt]lg1.base east) -- ++(16.6pt,0) coordinate (le2);
    \node (lg2) at ([xshift=3pt,yshift=-2pt]le2) {$3$};
    \draw[line width=4.2pt,pathearly,opacity=0.42,line cap=round]
      ([xshift=50pt,yshift=2pt]lg2.base east) -- ++(11.8pt,0) coordinate (le3);
    \node (lg3) at ([xshift=5.1pt,yshift=-2pt]le3) {$P_{2,1}$};
    \draw[line width=4.2pt,pathlate,opacity=0.5,line cap=round]
      ([xshift=\legendgap+2.1pt,yshift=2pt]lg3.base east) -- ++(11.8pt,0) coordinate (le4);
    \node (lg4) at ([xshift=5.1pt,yshift=-2pt]le4) {$P_{3,1}$};
  \end{scope}
\end{tikzpicture}
\caption{The skeleton in full, for $b=4$ branches, $\tau=3$ phases and $d=2$.  Here $H$ is
a four-cycle, with $\pi_1=(1\,2)(3\,4)$ and $\pi_2=(1\,3)(2\,4)$, and the words are the
$c^{\tau-1}=9$ pairs over $\{0,1,2\}$, in Gray code order.  Columns are the sets $C_w$, rows
the indices $a$.  Solid edges are the identity matchings; between any two consecutive
columns there is one further matching, dashed, which replaces the identity from phase~$2$
on if their words differ in the first entry, and from phase~$3$ on if they differ in the
second.  The paths of phase~$1$ are therefore the rows, and every target has at most two
neighbours on each side, so the skeleton has maximum degree four.  The shaded paths are
$P_{2,1}$ and $P_{3,1}$, those of branch~$1$ in phases $2$ and $3$.  Both start at
$\rho_1=x_{00,1}$; $P_{2,1}$ changes row only where the first entry of the words changes,
and $P_{3,1}$ between any two consecutive columns.}
\label{fig:skeleton}
\end{figure}
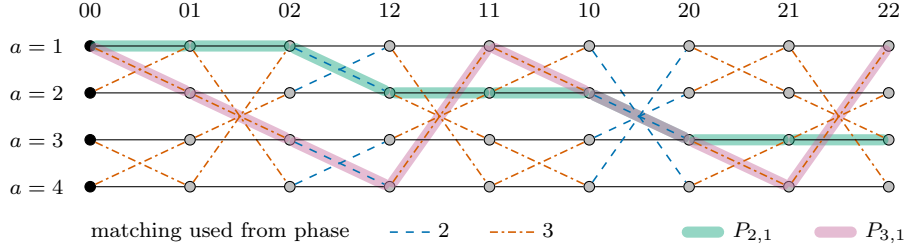

\subsection{Ordering the targets}\label{sec:skeleton}
We order the words, and with them the columns, by the reflected $c$-ary Gray code; see
Knuth~\cite[Section~7.2.1.1]{Knuth11} or Savage~\cite{Savage97} for background.  For a list $\mathcal L$ of words and a letter $\beta$, write
$\beta\mathcal L$ for the list obtained by prefixing $\beta$ to every word of $\mathcal L$, and
$\mathcal L^R$ for $\mathcal L$ in reverse order.  Let $\Gamma_0$ consist of the empty word alone, and for $\ell\ge0$ let
\[
  \Gamma_{\ell+1}=0\Gamma_\ell,\ 1\Gamma_\ell^R,\ 2\Gamma_\ell,\ 3\Gamma_\ell^R,\ \ldots,
\]
with one block for each letter $\beta\in\{0,\ldots,d\}$, the blocks using $\Gamma_\ell$ and
$\Gamma_\ell^R$ alternately; for $c=3$, $\Gamma_2$ is the order of the columns in
\cref{fig:skeleton}.  By induction on $\ell$, the list $\Gamma_\ell$ contains every word of
length $\ell$ exactly once, starts with the all-zero word, and any two of its consecutive
words differ in exactly one position.  Indeed, in $\Gamma_{\ell+1}$ the first block is
$0\Gamma_\ell$, consecutive words inside a block differ in one position because those of
$\Gamma_\ell$ and of its reverse do, and where two blocks meet the two words differ in the
first letter and have the same suffix, the last or the first word of $\Gamma_\ell$.  We use
$\Gamma_{\tau-1}$, which can be listed in time polynomial in its length $q=c^{\tau-1}$, and
write $w^{(1)},\ldots,w^{(q)}$ for its words in this order.

The path $P_{i,j}$ passes through the columns $C_{w^{(1)}},\ldots,C_{w^{(q)}}$ in this order
and takes from each the unique target whose label in phase $i$ is $j$, so its vertex set is
$X_{i,j}$ and each of its edges joins two consecutive columns.  The target
$\rho_j=x_{w^{(1)},j}$ has label $j$ in every phase, because $w^{(1)}$ is the all-zero word
and $\pi_0$ is the identity, so every path $P_{i,j}$ starts at $\rho_j$.

Consider two consecutive columns $C_w$ and $C_{w'}$.  In phase $i$ the paths cross from one to the other by pairing each target
of $C_w$ with the target of $C_{w'}$ that has the same label; this pairing is a perfect
matching between the two columns, which we call $M_i$.  The skeleton has between the two
columns exactly the edges of $M_1,\ldots,M_\tau$, so a target has at most as many neighbours
in the other column as there are different matchings among them.  Since $w$ and $w'$ differ
in only one position, there are at most two, as \cref{lem:twomatchings} shows.

\begin{lemma}\label{lem:twomatchings}
Let $w$ and $w'$ be consecutive words, differing in position $\mu$.  Then
$M_1=\cdots=M_\mu$ is the identity matching $\{x_{w,a}x_{w',a}:\ a\in[b]\}$, and
$M_{\mu+1}=\cdots=M_\tau$.
\end{lemma}

\begin{proof}
In phase~$1$ the label of every target is its row, so $M_1$ is the identity matching.  Now
let $1\le i<\tau$ with $i\ne\mu$.  Then $w_i=w'_i$, so by \eqref{eq:label} the labels of
both columns in phase $i+1$ are obtained from those in phase $i$ by the same permutation
$\pi=\pi_{w_i}$, and
\[
  \lambda_{i+1}(w,a)=\lambda_{i+1}(w',a')
  \iff \pi(\lambda_i(w,a))=\pi(\lambda_i(w',a'))
  \iff \lambda_i(w,a)=\lambda_i(w',a'),
\]
the last step because $\pi$ is a permutation.  Hence $M_{i+1}=M_i$.  Applying this for
$i<\mu$ and for $i>\mu$ gives $M_1=\cdots=M_\mu$ and $M_{\mu+1}=\cdots=M_\tau$.
\end{proof}

By \cref{lem:twomatchings}, the skeleton, the union of all the paths $P_{i,j}$, consists
of the $b$ rows together with at most one further perfect matching between any two
consecutive columns; \cref{fig:skeleton} shows it on a small example.  A target therefore has
at most two neighbours on each side, and the skeleton has maximum degree at most four.

\begin{proof}[Proof of \cref{thm:assignment}]
Take for $b_0$ the constant of \cref{lem:expander}, and let $b\ge b_0$ and $\tau\ge1$.  The
sets $X_{i,j}$ defined by the labels \eqref{eq:label} form a branch assignment of the $bq$
targets (\cref{sec:labels}), and it is evasive by \cref{lem:evasive}.  The paths $P_{i,j}$
form a path realisation with first vertices $\rho_1,\ldots,\rho_b$, and by
\cref{lem:twomatchings} their union, the skeleton, has maximum degree at most four.  Every
step is deterministic and takes time polynomial in $bq$: the graph $H$ and the permutations
$\pi_1,\ldots,\pi_d$ are given by \cref{lem:expander}, the labels by \eqref{eq:label}, and
the Gray code order can be listed in time polynomial in its length~$q$; the paths are then
read off from the labels and the order.
\end{proof}

\section{Proof of the main theorem}\label{sec:proof}

Let $d$ and $c$ be as in \cref{sec:construction}, and let the parameters be given by
\eqref{eq:params} and \eqref{eq:params2}.  Let $b_0$ be the constant of \cref{lem:expander},
and let $\Delta_0$ be the smallest power of two with $\Delta_0/64\ge\max\{8,b_0\}$.  We call
$\Delta$ \emph{admissible} if it is a power of two with $\Delta_0\le\Delta\le n-1$.  An
admissible $\Delta$ gives enough branches both for the construction and for
\cref{lem:expander}; in particular $\Delta\ge512$, the standing assumption of
\cref{sec:construction}.  The condition also forces $n>\Delta_0$, so nothing has to be
assumed separately about $n$.  The constant $\Delta_0$ is absolute but not explicit, since
$b_0$ comes from the expander construction we quote, and we do not optimise the
numerical constants.  \Cref{lem:params} collects the estimates we need, including the edge
count used in \cref{sec:intro}.

\begin{lemma}\label{lem:params}
Let $\Delta$ be admissible.  With these parameters, $b=\Delta/64\ge\max\{8,b_0\}$, $2\le p\le 2n/\Delta+1$,
$n/(3c)<bq\le n/3$, $n/2<k<n$ and $L=bk/8\ge k$.  Moreover
$\tau\ge(1+\log(n/\Delta))/(2\log c)$.  In particular
$p$ and $q$ are $\Theta(n/\Delta)$, $L=\Theta(\Delta n)$ and $\tau=\Theta(1+\log(n/\Delta))$.
Finally, $\und(\TG)$ has between $\Delta n/128$ and $2n+\Delta n/64$
edges.
\end{lemma}
\begin{proof}
Since $\Delta$ is admissible, $b=\Delta/64\ge\Delta_0/64\ge\max\{8,b_0\}$.  Since $\Delta\le n$ and $p=\lceil2n/\Delta\rceil$, we
have $2\le p\le 2n/\Delta+1$, hence
\[
  bp\le\frac{\Delta}{64}\Bigl(\frac{2n}{\Delta}+1\Bigr)=\frac n{32}+\frac \Delta{64}\le\frac{3n}{64}.
\]
The set in the definition of $\tau$ is not empty because $b\le n/64\le n/3$, and by the
maximality of $\tau$ we have $bc^{\tau}>n/3$.  Hence
\[
  \frac{n}{3c}<bq\le\frac n3,\qquad
  \frac{n}{3cb}<q\le\frac{n}{3b},\qquad\text{and so}\qquad q=\Theta\Bigl(\frac n\Delta\Bigr).
\]
The number of cycle vertices satisfies
\[
  k=n-bp-bq\ge n-\frac{3n}{64}-\frac n3>\frac n2,
\]
and $k<n$.  In particular $k\ge3$, and $L=bk/8\ge k\ge1$ because $b\ge8$.

For the number of phases, $c^{\tau}=cq>n/(3b)\ge n/\Delta$ because $3b\le \Delta$, so
$\tau>\log_c(n/\Delta)$, and $\tau\ge1$ by definition.  A maximum of two numbers is at least
their average, so
\[
  \tau\ge\max\Bigl\{1,\log_c\frac n\Delta\Bigr\}\ge\frac12\Bigl(1+\log_c\frac n\Delta\Bigr)
  \ge\frac{1}{2\log c}\Bigl(1+\log\frac n\Delta\Bigr),
\]
where the last step uses $\log c>1$.  In the other direction $c^{\tau-1}=q\le n/(3b)$ gives
$\tau\le1+\log_c(64n/(3\Delta))$, so $\tau=\Theta(1+\log(n/\Delta))$.

Finally, we count the edges of the underlying graph.  There are $k$ cycle edges and $bp$
edges on the gate paths and between gates and targets, and the skeleton has at most $2bq$
edges, since it has maximum degree at most four.  Between the cycle and the gates there are exactly
$bk$ edges, those joining each cycle vertex to the $b$ gates of its class; each of them appears, because every cycle vertex is the hub at some step of the
first phase, as $L\ge k$.  Since $k+bp+bq=n$ and $bq\le n/3$, the total is at most
$n+n/3+bk\le 2n+\Delta n/64$, and it is at least $bk\ge(\Delta/64)(n/2)=\Delta n/128$.
\end{proof}

\begin{proof}[Proof of \cref{thm:main}]
Assume first that $\Delta$ is admissible, and let $\TG$ be the temporal graph of
\cref{sec:construction}, built from the evasive branch assignment
and path realisation of \cref{thm:assignment}, which exist because $b\ge b_0$.  It has
exactly $n$ vertices, by the choice of $k$ (\cref{sec:vertices}).  
By \cref{prop:basic}, every snapshot of $\TG$ is a spanning tree in which the hub is the
only vertex of degree larger than three, so $\TG$ is always-connected, and the underlying
graph has maximum degree at most $\Delta/2+2\le\Delta$.

Consider any temporal walk of at most $\tau L$ time steps, from an arbitrary start vertex.
Extending it by waiting, we may assume that it has exactly $\tau L$ steps,
$W=(v_1,\ldots,v_{\tau L+1})$.  For every phase $i$ let $A_i\subseteq[b]$ be the set of
branches that $W$ does not visit during phase $i$.  By \cref{lem:half}, which applies
because $b\ge8$ and $L=bk/8$, the walk visits fewer than $b/2$ branches during phase $i$, so
$|A_i|>b/2$.  Since the assignment is evasive there is a target $x$ such that for every
phase $i$ there is $j_i\in A_i$ with $x\in X_{i,j_i}$.  During phase $i$ the target $x$
belongs to branch $j_i$, which $W$ does not visit, so $x$ is none of the positions
$v_{(i-1)L+1},\ldots,v_{iL+1}$.  These position ranges for $i=1,\ldots,\tau$ cover all
positions $v_1,\ldots,v_{\tau L+1}$.  Hence $W$ does not visit $x$, so no temporal walk of
at most $\tau L$ steps is an exploration, and $\texp(\TG,v^\ast)\ge\tau L+1$ for every
start vertex $v^\ast$.

By \cref{lem:params},
\[
  \frac{\Delta n}{1024}<L<\frac{\Delta n}{512}
  \qquad\text{and}\qquad
  \frac{1+\log(n/\Delta)}{2\log c}\le\tau=O\bigl(1+\log(n/\Delta)\bigr),
\]
so
\[
  \texp(\TG,v^\ast)>\tau L\ge\frac{\Delta n\,(1+\log(n/\Delta))}{2048\log c}
\]
for every start vertex $v^\ast$, and $\tau L=O(\Delta n\,(1+\log(n/\Delta)))$.  The temporal
graph is constant from step $\tau L$ on.  The assignment and its path realisation are
constructed deterministically in time polynomial in $bq\le n$ by \cref{thm:assignment}, and
each of the first $\tau L=O(n^2\log n)$ snapshots, which determine $\TG$, can then be
written down in time polynomial in $n$, so $\TG$ is constructed deterministically in time
polynomial in~$n$.

It remains to remove the assumption that $\Delta$ is a power of two.  Let $\Delta$ be any
integer with $\Delta_0\le\Delta\le n-1$, and let $\Delta'$ be the largest power of two with
$\Delta'\le\Delta$.  Then $\Delta/2<\Delta'\le\Delta$, and $\Delta'$ is admissible, because
$\Delta_0$ is itself a power of two at most $\Delta$ and hence $\Delta_0\le\Delta'$.  We use
the temporal graph of \cref{sec:construction} for $n$ and $\Delta'$.  Its underlying graph has
maximum degree at most $\Delta'\le\Delta$.  It is constant from some time step
$O(\Delta'n\,(1+\log(n/\Delta')))$ on, which is $O(\Delta n\,(1+\log(n/\Delta)))$ because
$\Delta'\le\Delta$ and $n/\Delta'<2n/\Delta$.  Its exploration time is at least
\[
  \frac{\Delta'n\,(1+\log(n/\Delta'))}{2048\log c}\ge\frac{\Delta n\,(1+\log(n/\Delta))}{4096\log c},
\]
because $\Delta'>\Delta/2$ and $n/\Delta'\ge n/\Delta$.  The other properties were
proved for every admissible value, in particular for $\Delta'$, and
\cref{thm:main} follows with $\gamma=1/(4096\log c)$.
\end{proof}

\begin{proof}[Proof of \cref{cor:problem}]
Every snapshot is a subgraph of the underlying graph, so every vertex $v$ has
$d_{\max}(v)\le\Delta$, and therefore $\overline{\Delta}\le\Delta$.  For the reverse inequality,
we use that the construction is run with a power of two $\Delta'>\Delta/2$, so that
$b=\Delta'/64$, and that $k>n/2$ by \cref{lem:params}.  The hub advances by one position at
every time step, so the cycle vertices are the hubs of the first $k$ time steps, which all
belong to the first phase because $L\ge k$, and each of them then has degree $b+1$.  Hence
$d_{\max}(v)\ge b+1$ for the $k$ cycle vertices, and
$\overline{\Delta}\ge(k/n)(b+1)>b/2=\Delta'/128>\Delta/256$; together with
$\overline{\Delta}\le\Delta$ this gives $\overline{\Delta}=\Theta(\Delta)$.

By \cref{thm:main} the exploration time is at least $\gamma\,\Delta n\,(1+\log(n/\Delta))$,
so it is enough to bound this quantity from below.  If $\Delta\le\sqrt n$ then
$\log(n/\Delta)\ge\frac12\log n$ and $\log\Delta\ge1$, so
\[
  \Delta n\,(1+\log(n/\Delta))\ge\frac12\,\Delta n\log n\ge\frac{\Delta}{2\log\Delta}\,n\log n,
\]
which gives both bounds of the corollary.  If $\Delta>\sqrt n$ then $\log n<2\log\Delta$, so
\[
  \Delta n\,(1+\log(n/\Delta))\ge\Delta n>\frac{\Delta}{2\log\Delta}\,n\log n.\qedhere
\]
\end{proof}

\section{Limits and variants of the construction}\label{sec:variants}
We show that our analysis is tight and that the scheme behind it cannot give more, and then
consider three variants, each changing one
ingredient: which degree is bounded, the hub, and the number of moves per step.  Throughout
this section $\Delta$ is admissible.

\paragraph*{Tightness for this construction.}
The bound of \cref{thm:main} is tight for our construction, and not only because $\TG$ is
constant after $\tau L$ steps, which lets the agent finish in $O(n)$ further steps.  Let
$\TG^\ast$ agree with $\TG$ on its first $\tau L$ steps and then repeat its $\tau$ phases for
ever, phase $i'$ using the paths $P_{i,j}$ of phase $i$ whenever $i'\equiv i\pmod\tau$.
Like $\TG$, it cannot be explored within $\tau L$ steps, and it has the same underlying
graph, so the lower bound of \cref{thm:main} holds for it as well.  On $\TG^\ast$ an agent
matches the lower bound by visiting, in each phase, every vertex of a
constant fraction of the branches, those carrying the most targets it has not visited yet.

\begin{proposition}\label{prop:tight}
For every start vertex $v^\ast$ we have
$\texp(\TG^\ast,v^\ast)=O(\Delta n\,(1+\log(n/\Delta)))$, so the bound of \cref{thm:main} is
best possible for $\TG^\ast$ up to a constant factor.
\end{proposition}
\begin{proof}
A \emph{sweep} of branch $j$ starts on the hub at some step $t$.  The walk takes the hub
edge to $g_{j,u(t)}$, follows the gate path to $g_{j,p}$ and the target path to its end, and
comes back along the target path and the whole gate path to $g_{j,1}$. 
There it waits until the hub lies in the class $S_1$,
which takes fewer than $p$ steps, and returns to the cycle through the hub edge.  In this
way it visits every gate of branch $j$ and every target of the branch in the current phase,
in at most $3p+2q$ steps, and $3p+2q<n/8<k/4$ by \cref{lem:params} and $\Delta\ge512$.  The
walk now stands on the old hub, at the far end of the cycle path.  It walks along the cycle
path towards the hub, which advances towards it, and meets the hub within $k/2+1$ steps;
when only one edge separates them it waits instead of moving, since a move would put it
behind the hub again.  A sweep and the return to the hub thus take fewer than $k$ steps.
Outside sweeps the walk stays on the hub by moving forward along the cycle, so if it is on
the hub when a phase begins, it can sweep any $\floor{b/8}\ge b/16$ branches of its choice
within that phase of $L=bk/8$ steps.

From an arbitrary start vertex the walk reaches the hub within $O(k)$ steps and stays on it
until the next phase begins, which costs at most one phase. In every phase it then sweeps the $\floor{b/8}$ branches that
carry the most targets it has not visited yet.  These carry at least a $1/16$ fraction of
the unvisited targets, so each phase multiplies the number of unvisited targets by at
most $15/16$.  After $O(1+\log q)$ phases at most $\floor{b/8}$ of the $bq$ targets remain
unvisited, in at most $\floor{b/8}$ branches, and the next phase visits them all.  The walk
then sweeps each of the $b$ branches once, $\floor{b/8}$ per phase, which visits every gate
within $O(L)$ steps, and finally moves forward along the cycle for $k-1$ steps, which keeps
it on the hub and visits every cycle vertex.  Altogether this takes
$O(L\,(1+\log q))=O(\Delta n\,(1+\log(n/\Delta)))$ steps, because $L=\Theta(\Delta n)$ and
$q=\Theta(n/\Delta)$ by \cref{lem:params}.
\end{proof}

\paragraph*{Limits of the scheme.}
However the targets are assigned, their number caps the number of phases.

\begin{proposition}\label{prop:phases}
Let $b$ be even.  Every evasive branch assignment with $b$ branches and $\tau$ phases has
more than $b\,2^{\tau-2}$ targets, even if the sets $X_{i,j}$ have different
sizes; in particular, one with at most $n$ targets has fewer than $\log_2(n/b)+2$ phases.
\end{proposition}
\begin{proof}
Suppose that there are at most $b\,2^{\tau-2}$ targets.  Consider an idealised agent that
can visit any $b/2$ branches in each phase, and let it choose them greedily as in the proof
of \cref{prop:tight}: those that contain the most targets it has not visited yet.  In each
of the first $\tau-1$ phases it then visits at least half of the unvisited targets, so at
most $b\,2^{\tau-2}/2^{\tau-1}=b/2$ remain for phase $\tau$, where they lie in at most $b/2$
branches and are all visited.  Let $A_i$ be the set of the $b/2$ branches not visited in
phase $i$.  Every target is visited in some phase $i$, and in that phase it lies in a
branch outside $A_i$; hence no target lies in a branch of $A_i$ in every phase, and the assignment is not
evasive.  The second claim follows, since an evasive assignment with at most $n$ targets
has $b\,2^{\tau-2}<n$.
\end{proof}

Our assignment has $bq\le n/3$ targets and $\tau=1+\lfloor\log_c(n/(3b))\rfloor$ phases by
\eqref{eq:params}, so it is within a constant factor of this ceiling for every $\Delta$; the
only loss is that it counts in base $c$ rather than in base two.  The length of a phase is
capped too.  From any vertex the agent reaches the hub within $O(k)$ steps, and a walk
through a whole branch and back to the hub then takes fewer than $k$ steps (\cref{prop:tight}), so
within $(b/2+1)k+O(k)$ steps of a phase it can visit more than half of the branches, and
evasion no longer applies.  Phases therefore have $O(bk)=O(bn)$ steps, and there are fewer
than $\log_2(n/b)+2$ of them.  Since the hub has degree $b+1\le\Delta$ and
$b\,(1+\log(n/b))=O(\Delta\,(1+\log(n/\Delta)))$ for $b<\Delta$, this scheme cannot give
more than $O(\Delta n\,(1+\log(n/\Delta)))$ steps, whatever the assignment, and \cref{thm:main}
attains this.

\paragraph*{Bounded snapshot degree only.}
If only the snapshot degree has to be bounded, the two devices that serve the underlying
degree can be dropped: the gates, which keep the number of cycle neighbours of every vertex
outside the cycle bounded, and the Gray code order of the targets, which gives the skeleton
bounded degree.  Attach the targets of branch $j$ in phase $i$ directly to the hub, as a path
through $X_{i,j}$ in an arbitrary order, and put the $bp$ vertices that were gates on the
cycle, so that there are still $n$ vertices.  Every $G_t$ is then a \emph{spider}, a tree
with a single vertex of degree larger than two, namely the hub, of degree $b+1$; hence
$\Delta_s=b+1=\Theta(\Delta)$, while the underlying degree is no longer bounded.  The
branches are still joined to each other only through the hub, so \cref{lem:half} applies
verbatim once ``gate'' is read as ``first vertex of the branch'' in the definitions of
entering and leaving a branch, and \cref{lem:evasive} is untouched.  The bound of
\cref{thm:main} therefore holds with $\Delta_s$ in place of $\Delta$, for temporal graphs
whose snapshots are spiders.

Since the skeleton no longer needs bounded degree, the assignment need not come from an
expander either: a random one is evasive with high probability.  Choose $q$ independently of
$\tau$, subject only to $bq\le n/3$, and assign every target to a branch independently and
uniformly at random in every phase.  The branches then hold $q$ targets only on average and
may even be empty, which spiders tolerate, and evasion is defined as in
\cref{def:evasive}.  It suffices to consider sets $A_1,\ldots,A_\tau\subseteq[b]$ with
$|A_i|=\lceil b/2\rceil$, since enlarging an $A_i$ can only help.  For fixed such sets, a
target lies in a branch of $A_i$ in every phase $i$ with probability at least $2^{-\tau}$,
independently of the other targets, so the probability that no target does is at most
$(1-2^{-\tau})^{bq}\le e^{-bq\,2^{-\tau}}$.  There are fewer than $2^{b\tau}$ choices of the
$A_i$, so the assignment is evasive with probability at least
$1-2^{b\tau}e^{-bq\,2^{-\tau}}$, which is at least $1-e^{-b\tau}$ as soon as
$q\ge2\tau\,2^{\tau}$.  With $bq\le n/3$ this allows, for large
$n/\Delta$, $\tau=\log_2(n/\Delta)-O(\log\log(n/\Delta))$ phases, within
$O(\log\log(n/\Delta))$ of the ceiling of \cref{prop:phases}, against the $\log_c(n/\Delta)$
phases of \cref{sec:assignment}.  
What the expander buys is therefore not evasion but an evasive assignment with the column
structure that the Gray code order needs, and a deterministic construction.

\paragraph*{The role of the hub.}
The factor $\Delta$ in \cref{thm:main} comes from the hub, the only vertex of large degree.
It is joined to $b=\Theta(\Delta)$ branches and advances by one position per step, so the
agent can enter a branch at most once every $\Theta(n)$ steps (\cref{lem:entries}), and in a
phase of $\Theta(\Delta n)$ steps it visits fewer than half of the branches
(\cref{lem:half}).  Without a hub, the combination of phases and an evasive assignment still works with two
branches, and gives the bound known for snapshots that are paths.  Let $n$ be a power of two, fix a set of $n/2$
vertices, the middle, and give the other $n/2$ vertices the $n/2$ binary codes of length
$\log_2(n/2)$, one each.  There are $\log_2(n/2)$ phases of $n/2$ steps, and every snapshot
of phase $i$ is the same path: it runs through the $n/4$ vertices whose code has $i$-th bit
$0$, then through the middle, and then through the other $n/4$ vertices.  The middle cannot
be crossed within a phase, so in each phase the agent visits vertices of at most one end,
and the vertex whose $i$-th bit points, for every $i$, to an end not visited in phase $i$ is
never visited; exploration therefore takes more than $(n/2)\log_2(n/2)$ steps. 
With paths that are fixed during each phase, no evasive assignment gives more than
$O(n\log n)$: the agent can then visit all $n$ vertices within $2n$ steps, so phases must
be shorter than $2n$, and \cref{prop:phases} allows only $\log_2n+O(1)$ of them.  The best known lower bound for
snapshots that are paths is $\Omega(n\log n)$~\cite{EHK21,AGMZ22}, and we do not know
whether it can be improved.

\paragraph*{Several moves per time step.}
The rotating cycle relies on the agent making one move per step.  With two moves per step,
as allowed by Erlebach et al.~\cite{EKLSS19}, an agent on the gate $g_{j,u(t)}$ at the start
of step $t$ can pass through the hub into any other branch within that step, so it need
never end a step behind the hub, and the delay of \cref{lem:ring} disappears.  From
$g_{j,1}$ the agent can walk to the end of the target path of branch $j$ and back in $p+q-1$
steps, wait fewer than $p$ steps until the hub lies in $S_1$, and pass to $g_{j+1,1}$; the
$b$ branches thus take fewer than $b\,(q+2p)<n/2$ steps by \cref{lem:params}.  From any
start vertex the agent reaches some $g_{j,1}$ within $O(k)$ steps and, waiting there for the
next phase if fewer than $n/2$ steps of the current one remain, visits every gate and every
target within a single phase, since $L>n/2$.  The cycle vertices take $O(k)$ further steps,
so the temporal graph of \cref{thm:main} is explored in $O(n)$ steps.  
We do not know whether any lower bound with a linear dependence on the degree survives two
moves per step.  The $\Omega(n\log n)$ lower bound of~\cite{EHK21} for planar underlying
graphs of maximum degree four, which has no such dependence, does survive: Erlebach
et al.~\cite[Corollary~5.2]{EKLSS19} show that it holds for any constant number of moves per
step.

\section{Open questions}\label{sec:conclusion}

\Cref{thm:main} answers Problem~3.1 of~\cite{BGMR25} positively, but the quantitative picture
is still incomplete.  For always-connected temporal graphs on $n$ vertices with underlying
maximum degree at most $\Delta$, the largest exploration time lies between our
${\Omega(\Delta n\,(1+\log(n/\Delta)))}$ and the $O(\min\{n^2,n^{3/2}\sqrt{\Delta\log n}\})$
that follows from~\cite{MS16,BGMR25}.  The two differ by a polynomial factor whenever
$\Delta$ is polynomially smaller than $n$, and closing this gap is the main question, open
even for constant degree: there the known bounds are $\Omega(n\log n)$~\cite{EHK21,AGMZ22}
and $O(n^{3/2}\sqrt{\log n})$~\cite{BGMR25}, and our bound does not improve the lower one.
Nor can our technique, as it stands, close the gap from below: the analysis is tight for the
graphs we build (\cref{prop:tight}), and no choice of expander or of assignment makes the
scheme give more than the bound of \cref{thm:main} (see the discussion after
\cref{prop:phases}).

Finally, two further questions concern the variants of \cref{sec:variants}: whether snapshots
that are paths allow a lower bound beyond $\Omega(n\log n)$, and whether a lower bound with a
linear dependence on the degree survives two moves per step.

\section*{Use of AI tools}

Large language models (ChatGPT and Claude) were used as research assistants throughout
this work, interactively and under our direction: in developing and checking the
arguments, in writing and revising the text, in producing the figures, and in writing code
to test the construction.  In particular,
the models suggested using Alon's explicit expanders, along with a Gray code order of the
targets, to keep the underlying degree bounded.  Every statement and every proof has been
verified by us, and we take full responsibility for the paper's content.

\bibliographystyle{plain}
\bibliography{references_arxiv}

\begin{table}[p]
\centering
\caption{Notation.}\label{tab:notation}
\smallskip
\setlength{\tabcolsep}{4pt}
\begin{tabular}{@{}lll@{}}
\toprule
symbol & meaning & defined in \\
\midrule
$n=|V|$ & number of vertices & \cref{sec:model} \\
$\TG=(G_t)_{t\ge1}$ & temporal graph, $G_t$ its snapshot at time $t$ & \cref{sec:model} \\
$\und(\TG)$ & underlying graph, the union of all snapshots & \cref{sec:model} \\
$W=(v_1,v_2,\ldots)$ & temporal walk, $v_t$ the position at the start of step $t$ & \cref{sec:model} \\
$\texp(\TG,v^\ast)$ & exploration time from $v^\ast$ & \cref{sec:model} \\
$d_{\max}(v)$ & largest degree of $v$ over all time steps & \cref{sec:model} \\
$\Delta$, $\Delta_s$, $\overline{\Delta}$ & underlying, snapshot, average temporal maximum degree & \cref{sec:model} \\
\addlinespace
$d=82$ & expander degree & \cref{sec:construction} \\
$c=d+1$ & number of permutations & \cref{sec:construction} \\
$b=\Delta/64$ & number of branches & \eqref{eq:params} \\
$\tau$ & number of phases & \eqref{eq:params} \\
$q=c^{\tau-1}$ & number of targets per branch & \eqref{eq:params} \\
$p=\lceil 2n/\Delta\rceil$ & number of classes, and of gates per branch & \eqref{eq:params2} \\
$k=n-bp-bq$ & number of cycle vertices & \eqref{eq:params2} \\
$L=bk/8$ & number of time steps per phase & \eqref{eq:params2} \\
$\gamma$, $\Delta_0$ & absolute constants of \cref{thm:main} & \cref{sec:proof} \\
admissible $\Delta$ & power of two with $\Delta_0\le\Delta\le n-1$ & \cref{sec:proof} \\
$\Delta'$ & largest power of two at most $\Delta$ & \cref{sec:proof} \\
\addlinespace
$r_1,\ldots,r_k$ & cycle vertices, with indices read modulo $k$ & \cref{sec:vertices} \\
$S_u$ & class of index $u$ & \cref{sec:vertices} \\
$g_{j,1},\ldots,g_{j,p}$ & gates of branch $j$; $g_{j,u}$ serves $S_u$ & \cref{sec:vertices} \\
$X$ & set of the $bq$ targets & \cref{sec:vertices} \\
\addlinespace
$X_{i,j}$ & targets of branch $j$ in phase $i$ & \cref{def:assignment} \\
$A_i$ & branches not visited in phase $i$, at least $b/2$ of them & \cref{def:evasive} \\
$P_{i,j}$, $\rho_j$ & target path of branch $j$ in phase $i$, and its first vertex & \cref{sec:assignment-abstract} \\
skeleton & union of all the paths $P_{i,j}$ & \cref{sec:assignment-abstract} \\
\addlinespace
$\hub_t=r_t$, $u(t)$ & hub at time $t$, and the index of its class & \cref{sec:steps} \\
$Q_t$ & cycle path at time $t$, read from the hub outwards & \cref{sec:steps} \\
branch $j$ in phase $i$ & vertex set $\{g_{j,1},\ldots,g_{j,p}\}\cup X_{i,j}$ & \cref{sec:steps} \\
visits branch $j$ & in phase $i$: some $v_t$ with $(i-1)L<t\le iL+1$ lies in it & \cref{sec:ring} \\
\addlinespace
$H$ & $d$-regular expander on the set $[b]$ of branches & \cref{lem:expander} \\
$\pi_0,\ldots,\pi_d$ & $\pi_0$ the identity, $\pi_1(x),\ldots,\pi_d(x)$ the neighbours of $x$ & \cref{sec:expander} \\
$b_0$ & threshold on $b$ in the expander construction & \cref{lem:expander} \\
$w$, $C_w$, $x_{w,a}$ & word, its column, and the column's target in row $a$ & \cref{sec:labels} \\
$\lambda_i(w,a)$ & label of $x_{w,a}$ in phase $i$, that is, its branch & \eqref{eq:label} \\
$N[Y]$ & $Y$ together with its neighbours in $H$ & \cref{lem:evasive} \\
$w^{(1)},\ldots,w^{(q)}$ & words in Gray code order & \cref{sec:skeleton} \\
$M_i$ & matching between two consecutive columns in phase $i$ & \cref{sec:skeleton} \\
\addlinespace
$\TG^\ast$ & $\TG$ with its $\tau$ phases repeated for ever & \cref{sec:variants} \\
sweep & walk through a whole branch and back to the cycle & \cref{prop:tight} \\
\bottomrule
\end{tabular}
\end{table}
\end{document}